\documentclass[pra,reprint]{revtex4-2}

\usepackage{tikz}
\usepackage{graphicx}
\usepackage[utf8]{inputenc}
\usepackage[english]{babel}
\usepackage{natbib}
\usepackage{amsmath,mathtools,amssymb,amsfonts,dsfont,amsthm,mathrsfs,bm}

\DeclarePairedDelimiterX{\ket}[1]{\lvert}{\rangle}{#1}
\DeclarePairedDelimiterX{\kket}[1]{\lvert}{\rangle\rangle}{#1}
\DeclarePairedDelimiterX{\bra}[1]{\langle}{\rvert}{#1}
\DeclarePairedDelimiterX{\bbra}[1]{\langle\langle}{\rvert}{#1}

\newcommand{\abs}[1]{\left\lvert #1 \right\rvert} 
\newcommand{\tr}{\operatorname{Tr}} 

\newcommand{\pnorm}[2][p]{\left\lVert #2 \right\rVert_{#1}} 
\newcommand{\dnorm}[1]{\lVert #1 \rVert_{\diamond}} 

\newcommand{\sepd}{\text{SepD}} 
\newcommand{\sepc}{\text{SepC}} 
\newcommand{\sepp}{\text{SEPP}} 
\newcommand{\fsepp}{\text{FSEPP}} 
\newcommand{\seppsc}{\text{SEPPSC}} 

\theoremstyle{definition}

\theoremstyle{definition}
\newtheorem{theorem}{Theorem}[section]

\newtheorem{proposition}[theorem]{Proposition}
\theoremstyle{remark}

\mathchardef\mhyphen="2D

\begin{document}
	\title{One-Shot Static Entanglement Cost of Bipartite Quantum Channels}
	\author{Ho-Joon Kim}\email{eneration@gmail.com}
	\affiliation{Department of Mathematics and Research Institute for Basic Sciences, Kyung Hee University, Seoul 02447, Korea}
	
	\author{Soojoon Lee}\email{level@khu.ac.kr}
	\affiliation{Department of Mathematics and Research Institute for Basic Sciences, Kyung Hee University, Seoul 02447, Korea}

	\begin{abstract}
		We first investigate the one-shot static entanglement cost to simulate a bipartite quantum channel under the set of non-entangling channels. The lower bound on the one-shot static entanglement cost is given by the generalized robustness of the target channel as well as the robust-generating power of the channel, which implies that the cost necessary to generate a bipartite quantum channel might not be retrievable in general. Next, we conceive a set of quantum channels that extends the set of non-entangling channels. We find out that the one-shot static entanglement cost under the extended set of channels is given by the channel's standard log-robustness. This quantity also gives the one-shot dynamic entanglement cost under a set of free superchannels that do not generate dynamic entanglement resource; the extended set is shown to be equivalent to the set of free superchannels.
	\end{abstract}

	\maketitle
	
\section{Introduction}
Physical resources in nature can show quantum features in two different forms: static resources in quantum states and dynamic resources in quantum dynamics. The prior one has been investigated intensely from the early days of quantum science to these days, focusing on the quantum entanglement that lies behind the most popular quantum paradox of the Schr\"odinger's cat \cite{schroedinger1935GegenwaertigeSituationQuantenmechanik,horodecki2009QuantumEntanglement}. The common theoretical pillar through all the static quantum resources has been understood in the framework of quantum resource theory of states \cite{chitambar2019QuantumResourceTheories}. Meanwhile, the dynamic quantum resources in quantum channels have been studied from the early days as well \cite{eisert2000OptimalLocalImplementation,nielsen2003QuantumDynamicsPhysical}, and it has regained interest in recent years as the general dynamic resource theory \cite{regula2020OneshotManipulationDynamical,xyuan2020OneshotDynamicalResource} as well as its applications concerning quantum computation \cite{xwang2019QuantifyingMagicQuantum-1,seddon2019QuantifyingMagicMultiqubit,takagi2020OptimalResourceCost}. In order to precisely analyze quantum dynamic resources \cite{gour2019HowQuantifyDynamical-1,yliu2020OperationalResourceTheory}, the concept of entropy \cite{gour2018EntropyQuantumChannel}, the coherence \cite{diaz2018UsingReusingCoherence,saxena2020DynamicalResourceTheory}, and the entanglement of quantum channels \cite{gour2019EntanglementBipartiteChannel,bauml2019ResourceTheoryEntanglement} have been investigated, and fault-tolerant quantum computation has been analyzed by applying resource theory of the magic states \cite{bravyi2005UniversalQuantumComputation,howard2017ApplicationResourceTheory,seddon2019QuantifyingMagicMultiqubit,xwang2019QuantifyingMagicQuantum-1,xwang2020EfficientlyComputableBounds}. These systematic investigations on quantum resources are rapidly unifying both static and dynamic quantum resources in a framework based on the fact that quantum states can be regarded as quantum channels with trivial input systems.

In fact, the two types of quantum resources are intimately intertwined; there have been various results on such relations from the beginning of quantum information science. Focusing on quantum entanglement, a static quantum entanglement resource, that is, an entangled bipartite state can be generated through given dynamic entanglement resources such as the CNOT gate or the SWAP gate since they enable to prepare maximally entangled states. Conversely, static entanglement resource can be used to simulate a quantum channel having a single input and output under local operations and classical communication (LOCC) \cite{berta2013EntanglementCostQuantum} and the positive partial transpose (PPT) channels \cite{xwang2018ExactEntanglementCost}. The static entanglement cost and distillable entanglement of a bipartite quantum channel were studied as well \cite{bauml2019ResourceTheoryEntanglement}. In spite of these results, however, it still remains to estimate the necessary amount of static entanglement resource to generate an arbitrary bipartite quantum channel in the one-shot scenario.

In this paper, we try to fill a gap in bridging the two types of quantum entanglement resources by looking into the one-shot static entanglement cost to simulate a bipartite quantum channel under resource non-generating channels. We furthermore investigate the one-shot static entanglement cost of a bipartite quantum channel allowing a larger set of free channels, which ends up in an exact value; it turns out that the extended set of free channels can be interpreted by exploiting dynamic resource theory of entanglement.

In Sec.~\ref{sec: res. theory of ent. and res. monotones}, we set up our notations and recapitulate static and dynamic resource theories of entanglement, and then introduce resource monotones that we use. Sec.~\ref{sec:one-shot cost under SEPP} introduces the one-shot static entanglement cost of a bipartite quantum channel under the set of non-entangling quantum channels; the quantity is shown to be lower-bounded by a resource monotone as well as the entanglement-generating power of the channel. Sec.~\ref{sec:one-shot cost under FSEPP} introduces a new class of free transformations, and the one-shot static entanglement cost under those channels is analyzed. Furthermore, it is argued why the one-shot static entanglement cost under the conceived channels is given by the same quantity that appeared in the one-shot dynamic entanglement cost under non-entangling superchannels. We conclude with Sec.~\ref{sec:conclusion} referring to the connection of our work to general transformations between quantum channels.

\section{Resource Theory of Entanglement and Resource Monotones}\label{sec: res. theory of ent. and res. monotones}
\subsection{Notation}
We use capital letters such as $ A $ and $ B $ to denote either physical systems or Hilbert spaces to describe the systems; Composite systems of system $A$ and $B$ is denoted as $AB$. Curly letters like $ \mathcal{N}_{A} $ for quantum channels which are the completely positive trace-preserving linear maps: the subscript represents the system that the channel acts on. When the input and output systems are different, it can be explicitly written as $ \mathcal{N}_{A\to B} $. The set of quantum channels from a system $ A $ to a system $ B $ is denoted as $ \mathcal{L}(A\!\to\! B) $, while $ \mathcal{L}(A) $ corresponds to that of the same input and output system $ A $. A bipartite quantum channel is a quantum channel that have two input and output systems, respectively, e.g., $ \mathcal{N}_{AB} $. Greek letters like $ \phi_{A} $, $ \psi_{B} $ denote density matrices of pure states of systems in the subscript, and $ \Phi_{AB}^{K} $ is the $ K $-maximally entangled state corresponding to $ \ket{\Phi^{K}}_{AB}=\frac{1}{\sqrt{K}}\sum_{i=0}^{K-1}\ket{ii}_{AB} $. The Choi state of a quantum channel $ \mathcal{N}_{A} $ will be denoted by $ J_{A\widetilde{A}}^{\mathcal{N}_{A}}\coloneqq \mathcal{I}_{A}\otimes \mathcal{N}_{\widetilde{A}}\left (\Phi_{A\widetilde{A}}^{\abs{A}}\right ) $, where $ \mathcal{I}_{A} $ is the identity channel, and $ \abs{A}$ is the dimension of the system $A$. We use the logarithm to base two.

\subsection{Resource theory of entanglement}
A resource theory consists of a set of free resources, either in the form of quantum states or in the form of quantum dynamics, and a set of free transformations that keeps the free resources \cite{chitambar2019QuantumResourceTheories}. The static resource theory of entanglement possesses a physically well-motivated set of free transformations related to locality in manipulating quantum systems, the LOCC channels \cite{horodecki2009QuantumEntanglement}. An LOCC channel consists of any local quantum operations and classical communications so that it allows to prepare free states termed as the separable states that can be written as a sum of product states as
\begin{equation}\label{key}
	\rho_{AB} = \sum_{i} p_{i} \phi_{A}^{(i)}\otimes \psi_{B}^{(i)},
\end{equation}
where $ p_{i}\ge 0 $ and $ \sum_{i}p_{i}=1 $ \cite{werner1989QuantumEPRHidden}. Although the set of LOCC channels is operationally intuitive, it is not a topologically closed set implying the existence of sequences of LOCC channels that do not converge to an LOCC channel \cite{chitambar2012}, incurring mathematical difficulties to manipulate the set fully \cite{chitambar2014}.

There have been several classes of bipartite quantum channels that include LOCC channels and are helpful to understand quantum entanglement \cite{regula2019OneshotEntanglementDistillation}. One of the substitutes for LOCC channels is the set of separable channels (SepC) characterized by their Choi states being separable states \cite{cirac2001}; it is strictly larger than the closure of the set of LOCC channels \cite{bennett1999,chitambar2009,chitambar2012}, is the largest set of bipartite quantum channels that are completely resource non-generating: a separable channel acting on local subsystems of a large system does not generate entanglement as a whole \cite{harrow2003RobustnessQuantumGates}. Meanwhile, the largest set of bipartite quantum channels that keep the set of separable states is the separability-preserving channels (SEPP) or non-entangling channels, which are by definition channels that send separable states to separable states \cite{harrow2003RobustnessQuantumGates}. A quantum state $\rho_{AB}$ is called a positive-partial-transpose (PPT) state \cite{peres1996} if $\rho_{AB}^{T_{A}}\ge 0$ where the superscript $T_{A}$ denotes the partial transpose map on the system $A$, a typical example of positive but not completely positive maps. A bipartite quantum channel is called a PPT channel \footnote{More precisely, it is called a completely-PPT-preserving channel.} if it sends PPT states to PPT states even if the channel acts on the subsystems \cite{rains1999BoundDistillableEntanglement,rains2001SemidefiniteProgramDistillable}; A PPT channel is also characterized by its Choi state being a PPT state.

 The static resource theory of entanglement extends to dynamic resource theories by regarding free channels as free resources and conceiving a set of free superchannels that do not generate resources. In general, a dynamic resource theory treats a set of quantum channels as free resources; a quantum state can also be seen as a quantum channel with a trivial input space $\mathbb{C}$. Transformations between quantum channels are described by supermaps that send linear maps to linear maps. A superchannel is a supermap that can be physically realized \cite{chiribella2008TransformingQuantumOperations,gour2019ComparisonQuantumChannels}; it should send completely positive maps to completely positive maps even when it acts on subsystems. It should preserve trace-preserving properties as well. It is known that any superchannel is equivalent to a pre-processing quantum channel followed by a post-processing quantum channel with an ancillary system.
When it comes to the dynamic resource theory of entanglement, one of the sets mentioned in the previous paragraphs such as LOCC, SEPP, PPT can be taken as a set of free dynamic resources, and superchannels that send free dynamic resources to themselves can be set as free superchannels. For instance, one can consider $\sepc(A\!:\!B)$ as free dynamic resources; A possible choice of free superchannels is the set of separability-preserving superchannels (SEPPSC) that send separable channels to separable channels \cite{hjkim2020OneshotManipulationEntanglement}; in Sec.~\ref{sec:one-shot cost under FSEPP}, we will use this dynamic resource theory of entanglement to interpret the static entanglement cost of a bipartite quantum channel under an extended set from SEPP. The dynamic resource theories of entanglement taking the LOCC channels as the free resources have first been
proposed \cite{gour2019EntanglementBipartiteChannel,gour2020DynamicalEntanglement}, and the dynamic resource theories of entanglement taking SEPP channels \cite{hjkim2020OneshotManipulationEntanglement} and PPT channels \cite{xwang2020CostQuantumEntanglement,xwang2018ExactEntanglementCost,gour2020DynamicalEntanglement,gour2019EntanglementBipartiteChannel,bauml2019ResourceTheoryEntanglement} have been established recently.

\subsection{Resource monotones}
We introduce two classes of resource monotones that can be adapted to quantify both the static and the dynamic quantum resources. Here we focus on the dynamic resource monotones. Let $ \mathbb{F} $ be the set of free channels. The generalized robustness of a quantum channel under the set $ \mathbb{F} $ is a resource monotone that has been well-investigated, and has operational meanings in tasks such as quantum state ensemble discrimination \cite{takagi2019GeneralResourceTheories}, resource erasure \cite{zwliu2017ResourceDestroyingMaps}, and one-shot catalytic dynamic entanglement cost \cite{hjkim2020OneshotManipulationEntanglement,regula2020OneshotManipulationDynamical}. In general, when the set of free resources is convex and closed, convex analysis provides useful tools to construct resource monotones \cite{regula2017ConvexGeometryQuantum,regula2020BenchmarkingOneshotDistillation}. Hereafter, we will consider $ \mathbb{F} $ as a set of free bipartite quantum channels.

The generalized robustness for a bipartite quantum channel $ \mathcal{N}_{AB} $ with respect to $ \mathbb{F} $ is defined as
\begin{align}\label{key}
	R_{\mathbb{F}}(\mathcal{N}_{AB}) &\coloneqq \min\{\lambda: \mathcal{N}_{AB}\le \lambda \mathcal{M}_{AB}, \mathcal{M}_{AB}\in \mathbb{F}\}\\
	&= \min\left \{\lambda: \dfrac{\mathcal{N}_{AB}+(\lambda - 1) \mathcal{M}_{AB}}{\lambda}\in \mathbb{F}\right \}.
\end{align}
From the definition, it is clear that the generalized robustness is the gauge function for the set of free resources \cite{rockafellar1996ConvexAnalysis,regula2017ConvexGeometryQuantum}. The generalized robustness is also related to the max relative entropy $D_{\max}$ of channels as given by
\begin{equation}\label{key}
	R_{\mathbb{F}}(\mathcal{N}_{AB}) = \min_{\mathcal{M}_{AB}\in \mathbb{F}} \exp\left (D_{\max}(\mathcal{N}_{AB}\Vert \mathcal{M}_{AB}) \right ),
\end{equation}
where $ D_{\max}(\mathcal{N}_{AB}\Vert \mathcal{M}_{AB} ) \coloneqq \log \min\{\lambda: \mathcal{N}_{AB}\le \lambda \mathcal{M}_{AB} \}$. The generalized log-robustness of a bipartite channel $ \mathcal{N}_{AB} $ is defined as $ LR_{\mathbb{F}}(\mathcal{N}_{AB})\coloneqq \log R_{\mathbb{F}}(\mathcal{N}_{AB}) $, and its smooth version
$LR_{\mathbb{F}}^{\varepsilon}$
with $ \varepsilon\ge 0 $ is given by
\begin{align}
	LR_{\mathbb{F}}^{\varepsilon}(\mathcal{N}_{AB}) &\coloneqq \min_{\mathcal{N}_{AB}'\approx_{\varepsilon} \mathcal{N}_{AB}} LR_{\mathbb{F}}(\mathcal{N}_{AB}')\\
	&=\min_{\mathcal{M}_{AB}\in\mathbb{F}}D_{\max}^{\varepsilon}(\mathcal{N}_{AB}\Vert \mathcal{M}_{AB}),
\end{align}
where $ \mathcal{N}_{AB}'\approx_{\varepsilon} \mathcal{N}_{AB} $ is a shorthand for the diamond-distance between channels \footnote{The diamond-distance between quantum channels is induced from the diamond norm which is defined for a map $\mathcal{E}_{A}$ as $\dnorm{\mathcal{E}_{A}}\coloneqq \max_{\psi_{AR}}\pnorm[1]{\mathcal{E}_{A}(\psi_{AR})}$ \cite{kitaev1997QuantumComputationsAlgorithms}. It has an operational meaning in the quantum channel discrimination task \cite{piani2009AllEntangledStates}.}, i.e., $\frac{1}{2}\dnorm{\mathcal{N}_{AB}'-\mathcal{N}_{AB}}\le \varepsilon  $, and the smooth max-relative entropy $D_{\max}^{\varepsilon}$ is given by $ D_{\max}^{\varepsilon}(\mathcal{N}_{AB}\Vert \mathcal{M}_{AB} ) = \min_{\mathcal{N}_{AB}'\approx_{\varepsilon} \mathcal{N}_{AB}} D_{\max}(\mathcal{N}_{AB}'\Vert \mathcal{M}_{AB} ) $.

Another class of resource monotone is the standard robustness of bipartite quantum channels with respect to $ \mathbb{F} $, which is defined as
\begin{align}
	R_{s, \mathbb{F}} (\mathcal{N}_{AB}) \coloneqq \min \bigg\{& \lambda: \dfrac{\mathcal{N}_{AB}+(\lambda - 1) \mathcal{M}_{AB}}{\lambda}\in \mathbb{F},\nonumber\\ &\mathcal{M}_{AB}\in \mathbb{F} \bigg\}.
\end{align}
 The standard log-robustness of a bipartite channel $ \mathcal{N}_{AB} $ is defined as $ LR_{s,\mathbb{F}}(\mathcal{N}_{AB})\coloneqq \log R_{s,\mathbb{F}}(\mathcal{N}_{AB}) $, and its smooth version is defined as
\begin{equation}\label{key}
	LR_{s,\mathbb{F}}^{\varepsilon}(\mathcal{N}_{AB}) \coloneqq \min_{\mathcal{N}_{AB}'\approx_{\varepsilon} \mathcal{N}_{AB}} LR_{s,\mathbb{F}}(\mathcal{N}_{AB}').
\end{equation}
The standard robustness of a bipartite channel has an operational meaning as the one-shot dynamic entanglement cost under the set of all separability-preserving superchannels (SEPPSC) which send separable channels to separable channels \cite{hjkim2020OneshotManipulationEntanglement}.
Note that a quantum state can be treated as a quantum channel with the trivial input space, that is, one-dimensional Hilbert space isomorphic to $ \mathbb{C} $ having one and only one quantum state $ 1 $. With this correspondence, the above quantities for quantum channels with a trivial input space reduce to those for quantum states.

\section{One-Shot Static Entanglement Cost under SEPP}\label{sec:one-shot cost under SEPP}
\begin{figure}[tbph]
	\centering
	\includegraphics{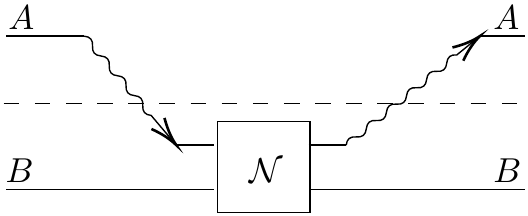}
	\caption{Simulation of a bipartite quantum channel $ \mathcal{N}_{AB} $ exploiting quantum teleportation (denoted via wavy arrows) two times using static entanglement resources.\label{fig: brute-force quantum channel simulation}}
	
\end{figure}
 In this section, we investigate what amount of the static entanglement resources is required to simulate a bipartite quantum channel utilizing SEPP channels which are the maximum resource non-generating channels in static entanglement resource theories. In principle, static entanglement resources in quantum states can be used to simulate any bipartite quantum channel under LOCC as depicted in Fig.~\ref{fig: brute-force quantum channel simulation} by exploiting quantum teleportation twice. This provides a trivial upper bound on the necessary amount of the static entanglement resources to simulate a bipartite quantum channel under LOCC or any other set of free channels which includes LOCC.
\begin{figure}[tbhp]
	\centering
	\includegraphics{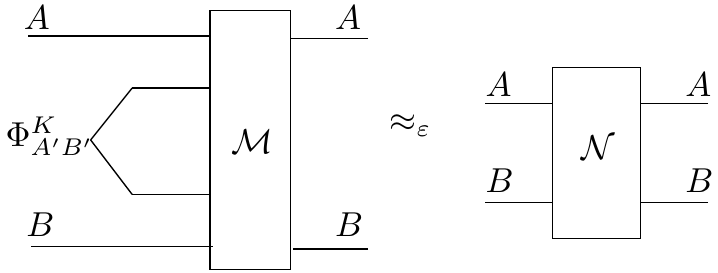}
	\caption{One-shot static entanglement cost of a bipartite channel $ \mathcal{N}_{AB} $ under a separability-preserving channel $ \mathcal{M}_{AA'BB'\to AB} $.}
	\label{fig: static entanglement cost of a channel under sepp}
\end{figure}

More precisely, the one-shot static entanglement cost of a bipartite quantum channel $ \mathcal{N}_{AB} $ under SEPP is an operational quantity that measures the minimum amount of the static entanglement resources $\Phi_{A'B'}^{K}$ to simulate a single instance of the quantum channel $ \mathcal{N}_{AB} $ under SEPP channels as depicted in Fig.~\ref{fig: static entanglement cost of a channel under sepp}. Formally it is defined as follows: Given $\varepsilon \ge 0$,
\begin{align}\label{key}
	&C_{\sepp}^{(1),\varepsilon} (\mathcal{N}_{AB})	\coloneqq \min \bigg\{ \log K : K\in\mathbb{N},\nonumber\\
	&\quad \dfrac{1}{2} \pnorm[\diamond]{\mathcal{N}_{AB} - \mathcal{M}_{AA'BB'\to AB}(\cdot \otimes \Phi_{A'B'}^{K})} \le \varepsilon, \nonumber\\
	&\quad  \mathcal{M}_{AA'BB'\to AB}\in \sepp(AA'\!:\!BB' \rightarrow A\!:\!B)\bigg\},
\end{align}
where $\sepp(AA'\!:\!BB' \rightarrow A\!:\!B)$ is the set of all SEPP channels that send separable states on $AA' \otimes BB'$ to separable states on $A \otimes B$. Thus, the one-shot static entanglement cost of a bipartite quantum channel $\mathcal{N}_{AB}$ is the minimum number of Bell states to approximate the channel under SEPP. Note that one can define the one-shot static entanglement cost using SEPP channels having the same input and output spaces, $AA'\otimes BB'$. It reduces to the above form after tracing out the local subsystems $ A'B' $, that is, 
\begin{equation}\label{key}
    \sepp(AA'\!:\!BB' \to A\!:\!B) = \tr_{A'B'} \sepp(AA'\!:\!BB'),
\end{equation} 
where $\sepp(X\!:\!Y) \coloneqq\sepp(X\!:\!Y\to X\!:\!Y)$.

We find a lower bound on this operational quantity as follows:
\begin{proposition}
	Given $ \varepsilon\ge 0 $, the one-shot static entanglement cost of a bipartite quantum channel $ \mathcal{N}_{AB} $ under SEPP is bounded below by the smooth generalized log-robustness of the channel with respect to SEPP, that is,
	\begin{equation}\label{key}
		LR_{\sepp}^{\varepsilon}(\mathcal{N}_{AB}) \le C_{\sepp}^{(1),\varepsilon} (\mathcal{N}_{AB}).
	\end{equation}
\end{proposition}
\begin{proof}
	Let $ C_{\sepp}^{(1), \varepsilon} (\mathcal{N}_{AB})= \log K $. There exists a quantum channel
	\begin{equation}
	    \mathcal{M}_{AA'BB'\to AB}\in \sepp(AA'\!:\!BB'\to A\!:\!B)    
	\end{equation}
	such that $ \mathcal{M}_{AA'BB'\to AB}(\cdot \otimes \Phi_{A'B'}^{K})\approx_{\varepsilon} \mathcal{N}_{AB} $. Using the dephasing channel $ \Delta_{A'}(X_{A'})=\sum_{i}\bra{i}X\ket{i}_{A'} \ket{i}\!\bra{i}_{A'} $, we have $ \Delta_{A'}\Phi_{A'B'}^{k} \in \sepd(A'\!:\!B')$, 
	where $\sepd(X\!:\!Y)$ is the set of all separable states on $X \otimes Y$. This leads to $ \mathcal{M}_{AA'BB'\to AB}(\cdot \otimes\Delta_{A'}\Phi_{A'B'}^{k})\in \sepp(A\!:\!B) $. We have that
	\begin{align}
		LR_{\sepp}^{\varepsilon}(\mathcal{N}_{AB})&\le \,LR_{\sepp} (\mathcal{M}(\cdot \otimes\Phi_{A'B'}^{K}))\nonumber\\
		&= \min_{\mathcal{F}_{AB}\in \sepp} D_{\max}(\mathcal{M}(\cdot \otimes\Phi_{A'B'}^{K})\Vert \mathcal{F}_{AB})\nonumber\\
		&\le\, D_{\max}(\mathcal{M}(\cdot \otimes\Phi_{A'B'}^{K})\Vert \mathcal{M}(\cdot \otimes\Delta_{A'}\Phi_{A'B'}^{K}))\nonumber\\
		&\le\, D_{\max} (\Phi_{A'B'}^{K} \Vert \Delta_{A'}\Phi_{A'B'}^{K})\nonumber\\
		&\le\, \log K,
		\label{eq:subscript}
	\end{align}
	where the subscript in $\mathcal{M}_{AA'BB'\to AB}$ is suppressed for readability.	This completes the proof.
\end{proof}

Next, we show that the static entanglement resource necessary to simulate a bipartite quantum channel is always greater than or equal to static entanglement resources generated from separable states. The maximum static entanglement that a bipartite quantum channel can generate from a separable state can be quantified by
\begin{equation}\label{key}
	P( \mathcal{N}_{AB} ) \coloneqq \max_{\sigma_{AB}\in\sepd(A\!:\!B)} R_{\sepp}\left( \mathcal{N}_{AB}(\sigma_{AB}) \right),
\end{equation}
which is called the robustness-generating power \cite{zanardi2000,takagi2019GeneralResourceTheories} of a bipartite quantum channel $ \mathcal{N}_{AB} $ \footnote{One can define the quantity with maximization over $ \sepd(AA'\!:\!BB') $.}.
Its smooth version is defined as $ P^{\varepsilon}(\mathcal{N}_{AB}) = \min_{\mathcal{N}_{AB}'\approx_{\varepsilon} \mathcal{N}_{AB}} P(\mathcal{N}_{AB}')$.
We find that the robust-generating power of a channel is no more than the one-shot static entanglement cost to simulate the channel as follows:
\begin{proposition}
	Given $ \varepsilon\ge 0 $ and a bipartite quantum channel $ \mathcal{N}_{AB} $, the following inequality holds:
	\begin{equation}\label{key}
		\log P^{\varepsilon}( \mathcal{N}_{AB} ) \le C_{\sepp}^{(1), \varepsilon} (\mathcal{N}_{AB}).
	\end{equation}
\end{proposition}
\begin{proof}
	Let $ \mathcal{M}_{AA'BB'\to AB} $ be a SEPP channel that simulates the quantum channel $ \mathcal{N}_{AB} $ with the static entanglement resource $ \Phi_{A'B'}^{K} $ such that
	\begin{equation}
	    \mathcal{N}_{AB}^{\varepsilon} \coloneqq \mathcal{M}_{AA'BB'\to AB}\left( \cdot \otimes \Phi_{A'B'}^{K} \right) \approx_{\varepsilon} N_{AB}.
	\end{equation}

	For any separable state $ \sigma_{AB}\in\sepd(A\!:\!B) $, we have that
	\begin{align}\label{key}
		R_{\sepp}\left( \mathcal{N}_{AB}^{\varepsilon}(\sigma_{AB}) \right) &= R_{\sepp}\left(  \mathcal{M}\left( \sigma_{AB}\otimes \Phi_{A'B'}^{K} \right) \right)\nonumber\\
		&\le R_{\sepp}\left( \sigma_{AB}\otimes \Phi_{A'B'}^{K} \right)\nonumber\\
		&= R_{\sepp}(\Phi_{A'B'}^{K})\nonumber\\
		&=K,
	\end{align}
	where the subscript in $\mathcal{M}_{AA'BB'\to AB}$ is suppressed for readability as in (\ref{eq:subscript}).	Therefore, it follows that
	\begin{align}
		C_{\sepp}^{(1), \varepsilon} (\mathcal{N}_{AB})&=\log K\nonumber\\
		&\ge \log \max_{\sigma_{AB}\in\sepd(A\!:\!B)} R_{\sepp}\left( \mathcal{N}_{AB}^{\varepsilon}(\sigma_{AB}) \right)\nonumber\\
		&=\log P( \mathcal{N}_{AB}^{\varepsilon} )\nonumber\\
		&\ge \log P^{\varepsilon}( \mathcal{N}_{AB}).
	\end{align}
	This completes the proof.
\end{proof}
The above result implies that the entanglement capacity or output static entanglement of a bipartite channel cannot be larger than the one-shot static entanglement cost necessary to simulate the channel in general.

\section{One-Shot Static Entanglement Cost under extended SEPP}\label{sec:one-shot cost under FSEPP}
In this section, we introduce a set of free channels larger than SEPP, and then calculate the one-shot static entanglement cost of a bipartite channel under the set. Recall that a multipartite quantum state $ \rho_{A_{1}\cdots A_{m}} $ is called a fully separable state if it can be written as a convex sum of product states as follows \cite{dur1999SeparabilityDistillabilityMultiparticle}:
\begin{equation}\label{key}
	\rho_{A_{1}\cdots A_{m}}=\sum_{j}p_{j}\rho_{A_{1}}^{(j)}\otimes \cdots \otimes \rho_{A_{m}}^{(j)}.
\end{equation}
We call a quantum channel sending a composite system $ A_{1}\dots A_{m} $ to a composite system $ B_{1}\dots B_{n} $ fully separability-preserving (FSEPP) if it sends a fully separable state to a fully separable state; the set of such channels is denoted as $ \fsepp(A_{1}\!:\!\cdots\!:\!A_{m}\rightarrow B_{1}\!:\!\cdots\!:\!B_{n}) $. When there are only two subsystems $ A $ and $ B $, we have the equality $ \fsepp(A\!:\!B) = \sepp(A\!:\!B) $.

The one-shot static entanglement cost of a bipartite quantum channel $ \mathcal{N}_{AB} $ under FSEPP, $C_{\fsepp}^{(1),\varepsilon}(\mathcal{N}_{AB})$, is an operational quantity defined as
\begin{align}
	&C_{\fsepp}^{(1),\varepsilon}(\mathcal{N}_{AB})\coloneqq \min \bigg\{ \log K : K\in\mathbb{N},\nonumber\\
	&\quad\dfrac{1}{2} \pnorm[\diamond]{\mathcal{N}_{AB} - \mathcal{M}_{AA'BB'\to AB}(\cdot \otimes \Phi_{A'B'}^{K})} \le \varepsilon,\nonumber\\
	&\quad\mathcal{M}_{AA'BB'\to AB}\in \fsepp(A\!:\!A'\!:\!B\!:\!B'\rightarrow A\!:\!B)\bigg\}.
\end{align}
Note that here we allow a set of free channels which is larger than the set $ \sepp(AA'\!:\!BB' \rightarrow A\!:\!B) $ in the previous section, since $ \mathcal{M}_{AA'BB'\to AB} $ is only required to send a (strict) subset of $ \sepd(AA'\!:\!BB') $ to the same set $ \sepd(A\!:\!B) $.
The main result of our work is as follows:
\begin{theorem}\label{thm: static ent. cost under fsepp}
	Given $ \varepsilon \ge 0 $, the one-shot static entanglement cost of a bipartite quantum channel $ \mathcal{N}_{AB} $ under FSEPP is given by
	\begin{equation}\label{key}
		C_{\fsepp}^{(1), \varepsilon} (\mathcal{N}_{AB})  = LR_{s,\sepp}^{\varepsilon}(\mathcal{N}_{AB}).
	\end{equation}
\end{theorem}
\begin{proof}
	Note that $ R_{s,\fsepp}^{\varepsilon}(\mathcal{N}_{AB})=R_{s,\sepp}^{\varepsilon}(\mathcal{N}_{AB}) $ due to $\fsepp(A\!:\!B) = \sepp(A\!:\!B)$. Let $ K = R_{s,\sepp}^{\varepsilon}(\mathcal{N}_{AB}) $. There exist quantum channels $ \mathcal{N}_{AB}'\approx_{\varepsilon} \mathcal{N}_{AB} $ and $ \mathcal{N}_{AB}'' \in \sepp(A\!:\!B)$ such that
	\begin{equation}\label{key}
		\mathcal{M}_{AB}^{\ast} = \dfrac{1}{K}\mathcal{N}_{AB}' + \left( 1-\dfrac{1}{K} \right)\mathcal{N}_{AB}''\in \sepp(A\!:\!B).
	\end{equation}
	We construct an $ \varepsilon $-simulating channel $ \mathcal{M}_{AA'BB'\to AB} $ as follows:
	\begin{align}
		&\mathcal{M}_{AA'BB'\to AB} \left( \rho_{AA'BB'} \right)\nonumber\\
		&\coloneqq\, \mathcal{N}_{AB}' \left(  \tr_{A'B'} \left( \Phi_{A'B'}^{K}\rho_{AA'BB'}\right)\right )\nonumber\\
		&\quad+\mathcal{N}_{AB}'' \left(  \tr_{A'B'}\left\{ (I_{A'B'}-\Phi_{A'B'}^{K})\rho_{AA'BB'}\right\}\right )\nonumber\\
		&=\,\tr \left( \Phi_{A'B'}^{K}\rho_{A'B'} \right) \mathcal{N}_{AB}'(\rho_{AB}')\nonumber\\
		&\quad +\tr \left\{ (I_{A'B'}-\Phi_{A'B'}^{K}) \rho_{A'B'} \right\} \mathcal{N}_{AB}''(\rho_{AB}''),
	\end{align}
	where $ \rho_{AB}' $ and $ \rho_{AB}'' $ are the post-measurement states of $ \rho_{AA'BB'} $ depending on the outcomes of the measurement $ \{\Phi_{A'B'}^{K}, I_{A'B'}-\Phi_{A'B'}^{K} \} $.
	Apparently, we have that
	\begin{equation}
	    \mathcal{M}_{AA'BB'\to AB} \left( \rho_{AB}\otimes \Phi_{A'B'}^{K} \right) = \mathcal{N}_{AB}'(\rho_{AB}).
	\end{equation}
	Furthermore, for $ \sigma_{AB}\in\sepd(A\!:\!B)$ and $\sigma_{A'B'}'\in\sepd(A'\!:\!B') $, it follows that
	\begin{align}
		&\mathcal{M}_{AA'BB'\to AB} (\sigma_{AB}\otimes \sigma_{A'B'}' )\nonumber\\
		&=\, \tr \left( \Phi_{A'B'}^{K}\sigma_{A'B'}' \right) \mathcal{N}_{AB}'(\sigma_{AB})\nonumber\\
		&\quad +\tr \left\{ (I_{A'B'}-\Phi_{A'B'}^{K}) \sigma_{A'B'}' \right\} \mathcal{N}_{AB}''(\sigma_{AB})\nonumber\\
	&= \,q\mathcal{M}_{AB}^{\ast}(\sigma_{AB}) 
	+(1-q)\mathcal{N}_{AB}''(\sigma_{AB})\nonumber\\
		&\in\, \sepd(A\!:\!B),
	\end{align}
	where $ q = K \tr \left( \Phi_{A'B'}^{K}\sigma_{A'B'}' \right)\le 1 $ because $\sigma'_{A'B'}$ is separable~\cite{horodecki1999general}. Since any fully separable state $ \sigma_{AA'BB'} $ can be written as a convex sum of product states of separable states such as $ \sigma_{AB}\otimes \sigma_{A'B'}' $ for $ \sigma_{AB}\in\sepd(A\!:\!B)$ and $\sigma_{A'B'}'\in\sepd(A'\!:\!B') $, the above result proves that
	\begin{equation}
	    \mathcal{M}_{AA'BB'\to AB}\in \fsepp(A\!:\!A'\!:\!B\!:\!B'\rightarrow A\!:\!B).
	\end{equation}
	
	The lower bound follows from the monotonicity of the standard log-robustness with respect to the compositions with free channels. Let $ C_{\fsepp}^{(1), \varepsilon} (\mathcal{N}_{AB}) = \log K $. There exists $ \mathcal{M}_{AA'BB'\to AB}\in \fsepp(A\!:\!A'\!:\!B\!:\!B'\to A\!:\!B)  $ such that $ \mathcal{M}_{AA'BB'\to AB}(\cdot \otimes \Phi_{A'B'}^{K})\approx_{\varepsilon} \mathcal{N}_{AB} $. It follows that
	\begin{align}
		LR_{s, \sepp}^{\varepsilon} (\mathcal{N}_{AB}) &=LR_{s, \fsepp}^{\varepsilon} (\mathcal{N}_{AB})\nonumber\\
		&\le LR_{s, \fsepp} (\mathcal{M}
		(\cdot \otimes\Phi_{A'B'}^{K}))\nonumber\\
		&\le LR_{s, \fsepp} (\cdot \otimes\Phi_{A'B'}^{K})\nonumber\\
		&\le \log K\nonumber\\
		&= C_{\fsepp}^{(1), \varepsilon} (\mathcal{N}_{AB}),
	\end{align}
	where $\sepp(A\!:\!B)=\fsepp(A\!:\!B)$ is used in the first line. This completes the proof.
\end{proof}

The above result can be alternatively understood in a dynamic entanglement resource theory. Taking the separable channels as the free resources, a superchannel is called separability-preserving superchannel (SEPPSC) if it sends separable channels to separable channels \cite{hjkim2020OneshotManipulationEntanglement}. Consider a set of separable channels with the trivial input space $ \sepc(\mathbb{C}\!\to\! A\!:\!B) $, which is isomorphic to the set of separable states $ \sepd(A\!:\!B) $, and a superchannel $ \Theta: \mathcal{L}(\mathbb{C}\!\to\! A'B')\!\to\! \mathcal{L}(AB) $ with the specified input and output space. The simulation of a quantum channel in previous paragraphs can be seen as a transformation from 
the input channels, here corresponding to $ \mathcal{L}(\mathbb{C}\to A'B') $, to the output channels, which is equivalent to a quantum channel in $ \mathcal{L}(AB) $ as depicted in Fig. \ref{fig: static cost as dyn. cost} \cite{chiribella2008TransformingQuantumOperations,gour2019ComparisonQuantumChannels}.

\begin{figure}[tbhp]
	\centering
	\includegraphics[width=\linewidth]{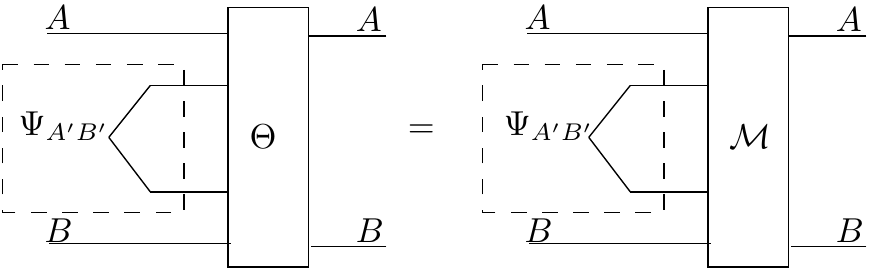}
	\caption{A superchannel $ \Theta:\mathcal{L}(\mathbb{C}\to A'B')\to \mathcal{L}(AB) $ is equivalent to a quantum channel $ \mathcal{M}\in \mathcal{L}(AA'BB'\to AB) $. The dashed rectangle encloses the input channel $ \Psi_{A'B'} $.\label{fig: static cost as dyn. cost}}
\end{figure}
Now we argue that $ \seppsc(\mathcal{L}(\mathbb{C}\to A'B')\to \mathcal{L}(AB)) $ is equivalent to $ \fsepp(A\!:\!A'\!:\!B\!:\!B'\rightarrow A\!:\!B) $. Consider a superchannel $\Theta\in\seppsc(\mathcal{L}(\mathbb{C}\to A'B')\to \mathcal{L}(AB))$. Since a superchannel consists of a pre- and post-quantum channels connected by an ancillary system \cite{chiribella2008TransformingQuantumOperations}, $ \Theta $ is implemented by a quantum channel $\mathcal{M}_{AA'BB'\rightarrow AB}\in \mathcal{L}(AA'BB'\to AB)$. Thus, the condition for $ \Theta $ being SEPPSC is, by definition, given by the condition for the quantum channel $\mathcal{M}_{AA'BB'\rightarrow AB} $ satisfying
\begin{equation}\label{key}
	\mathcal{M}_{AA'BB'\rightarrow AB}\left( \sigma_{AB}\otimes \sigma_{A'B'} \right)\in\sepd(A\!:\!B)
\end{equation}
for all $ \sigma_{AB}\in \sepd(A\!:\!B) $ and $ \sigma_{A'B'}\in \sepd(A'\!:\!B') $. As noted in the proof of Theorem \ref{thm: static ent. cost under fsepp}, this is equivalent to that $\mathcal{M}_{AA'BB'\rightarrow AB}\in \fsepp(A\!:\!A'\!:\!B\!:\!B'\rightarrow A\!:\!B) $. Hence, we can interpret the one-shot static entanglement cost of a quantum channel under FSEPP as the one-shot dynamic entanglement cost of a quantum channel under SEPPSC by regarding the static entanglement resource as a dynamic entanglement resource with a trivial input space.

\section{Conclusion}\label{sec:conclusion}
We have presented a lower bound on the one-shot static entanglement cost of a bipartite quantum channel under the set of SEPP channels given by the generalized log-robustness of the channel. The maximum static entanglement that can be generated by a bipartite quantum channel from a separable state has also been shown to be less than or equal to the one-shot static entanglement cost to simulate the channel in general.

We have defined the FSEPP channels, an extension of SEPP channels. The one-shot static entanglement cost of a bipartite quantum channel under the set of FSEPP channels is given by the standard log-robustness of the target channel. The latter also gives the one-shot dynamic entanglement cost under SEPPSC, the set of dynamic resource non-generating superchannels. One can understand this coincidence by regarding the static entanglement resource used to generate the target channel as a dynamic entanglement resource with trivial inputs; then the free quantum channel in the setting corresponds to the free superchannel in the dynamic resource theory. Since the free superchannel consists of a single quantum channel due to the trivial input space, the set of the free superchannel coincides with the set of FSEPP. This explains the reason why the one-shot static entanglement cost under FSEPP should be given by the standard log-robustness of the target channel, matching the result of the dynamic entanglement theory under SEPPSC \cite{hjkim2020OneshotManipulationEntanglement}.

Dynamic resource theories possess intricate structures that are absent in the static resource theories, although they share many facets analogous to those of the latter. An intriguing complication arises when one considers quantum channel transformations between many copies \cite{kfang2020NogoTheoremsQuantum,regula2020OneshotManipulationDynamical,regula2020FundamentalLimitationsQuantum}; With fixed causal order, $n$-copies of quantum channels can be combined in parallel, sequential, or adaptive ways where superchannels form a quantum comb \cite{chiribella2008QuantumCircuitArchitecture,chiribella2009TheoreticalFrameworkQuantum}. It is even possible to use $n$-copies of quantum channels without definite causal order in a quantum process \cite{chiribella2013QuantumComputationsDefinite}. These possibilities raise mathematically challenging problems as well as questions concerning their physical meanings. When it comes to the dynamic resource theory of entanglement, a few pioneering works already exist tackling the most general transformations \cite{gour2019EntanglementBipartiteChannel,bauml2019ResourceTheoryEntanglement}. Our work avoids such issues focusing on the one-shot scenario dealing with a single-copy transformations though, it already shows a subtle relation between the static and the dynamic entanglement resources as shown in the previous section; We wish that this work is helpful to the exploration of various dynamic resources.

\begin{acknowledgements}
H.-J. Kim acknowledges discussions with M. Plenio, L. Lami, and T. Theurer at Ulm university; L. Lami helped us clarify the definition of FSEPP. This research was supported by the National Research Foundation of Korea (NRF) grant funded by the Ministry of Science and ICT (MSIT) (Grant no.\ NRF-2019R1A2C1006337) and (Grant no.\ NRF-2020M3E4A1079678). S.L.\ acknowledges support from the MSIT, Korea, under the Information Technology Research Center support program (IITP-2021-2018-0-01402) supervised by the Institute for Information \& Communications Technology Planning \& Evaluation, and the Quantum Information Science and Technologies program of the NRF funded by the MSIT (No. 2020M3H3A1105796).
\end{acknowledgements}

\bibliography{[ref]_quantum_information}
\end{document}